\documentclass{article}
\pdfoutput=1

\usepackage{hyperref}
\usepackage{amsmath}
\usepackage{amssymb}
\usepackage{amsthm}
\usepackage{adjustbox}
\usepackage{bm}
\usepackage{pgfplots}
\usepackage{siunitx}
\usepackage{subcaption}
\usepackage{tikz}
\usepackage{tikz-cd}
\usetikzlibrary{graphs}
\usetikzlibrary{quotes}
\usetikzlibrary{decorations.pathreplacing}
\usepackage{algorithm}
\usepackage[noend]{algpseudocode}
\usepackage[capitalize]{cleveref}

\DeclareMathOperator*{\argmax}{argmax}

\DeclareMathOperator*{\orient}{orient}
\DeclareMathOperator*{\reduce}{reduce}

\DeclareMathOperator*{\op}{op}
\DeclareMathOperator*{\frt}{frt}
\DeclareMathOperator*{\rt}{rt}

\newcounter{Halgorithmic}
\AtBeginEnvironment{algorithmic}{\stepcounter{Halgorithmic}}
\ExpandArgs{c}\renewcommand{theHALG@line}{\arabic{Halgorithmic}.\arabic{ALG@line}}

\numberwithin{algorithm}{section}
\numberwithin{equation}{section}
\numberwithin{figure}{section}

\newenvironment{proofs}{%
  \proof}{\endproof}

\newtheorem{hyp}{Hypothesis}
\newtheorem{theorem}{Theorem}
\newtheorem{lemma}{Lemma}

\theoremstyle{definition}
\newtheorem{example}{Example}

\pgfplotstableset{col sep=comma}
\pgfplotsset{compat=1.7}

\title{Quickhull is Usually Forward Stable}
\author{Thomas Koopman \and Sven-Bodo Scholz}

\begin{document}

\maketitle

\begin{abstract}
    Quickhull is an algorithm for computing the convex hull of points
    in a plane that performs well in practice, but has poor complexity
    on adversarial input. In this paper we show the same holds for the
    numerical stability of Quickhull.
\end{abstract}

\section{Introduction}

The convex hull $CH(P)$ of a bounded set $P \subseteq \mathbb{R}^2$ is the 
smallest set that contains $P$ and is closed under taking convex combinations.
On a computer, we can only deal with finite sets $P$ and the convex hull
of a finite set is always a polygon. The planar convex hull problem is to list
the vertices of this polygon in clockwise (or counter-clockwise) order.
By abuse of notation, we will write $|CH(P)|$ for the number of vertices
on the convex hull.

% Previous work and research gap
The goal in algorithm design is to compute a good solution quickly. There are
many convex hull algorithms to choose from, differing in worst-case runtime
complexity and practical performance. The Quickhull algorithm is most performant
on many data sets, usually running in $O(|P| \log |CH(P)|)$, despite a
worst-case complexity of $O(|P|^2)$. The goal is to compute a good solution and 
not \textbf{the} solution because the convex hull is defined on values in 
$\mathbb{R}$.
This introduces uncertainties arising from measurements and rounding errors. 
Dealing with these uncertainties is a well-established field in mathematics, 
but has not been applied to Quickhull yet. This paper seeks to close that gap.

Using some jargon---which we explain assuming no prior knowledge in 
Subsection~\ref{subsec:robustness}---this paper provides the following 
contributions.

\begin{itemize}
    \item We define a measure of inaccuracy for the convex hull problem and do
          the corresponding perturbation analysis.
    \item We show that the Quickhull algorithm has a forward error bound of 
          $O(|P|^2 \epsilon)$ for $\epsilon$ the machine precision.
    \item We give a probabilistic argument that in practice the error bound is 
          more likely to grow proportionally to $\log(|CH(P)|)\epsilon$.
    \item We propose two variations of Quickhull that reduce the worst-case 
          error bound to
          $O(\sqrt{|P|}\log(|P|)\epsilon)$ or $O((\log |P|)^2\epsilon)$ for 
          cases where the runtime of Quickhull is $O(|P| \log |P|)$.
\end{itemize}

\section{Background}

\subsection{Quickhull}

The Quickhull algorithm \cite{Barber96} makes use of two geometric observations
illustrated by Figure~\ref{fig:quickhull_corr}. First, if $p$, $q$ are in the 
convex hull, then any point $r$ with maximum distance to the line through $p$
and $q$ is in the convex hull as well. Second, any point within the 
triangle $\Delta prq$ is not in the convex hull. So if the points $prq$ are 
listed in clockwise order, we only need to consider the points 
to the left of the oriented line segment $pr$ and to the left of the oriented
line segment $rq$. This idea can be applied recursively to obtain 
Algorithm~\ref{alg:quickhull_corr}.

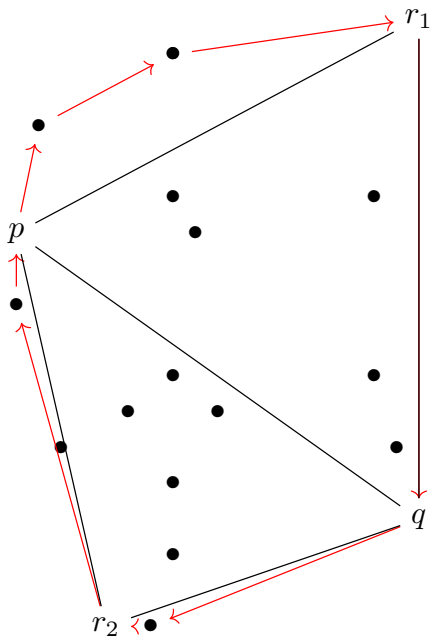
\begin{figure}[ht]
    \resizebox{0.5 \columnwidth}{!}{%
        \begin{tikzpicture}
            \node (6_2) at (1.500, 0.800) {$\bullet$};
            \node (16_9) at (4.000, 3.600) {$\bullet$};
            \node (7_9) at (1.750, 3.600) {$\bullet$};
            \node (7_18) at (1.750, 7.200) {$\bullet$};
            \node (9_8) at (2.250, 3.200) {$\bullet$};
            \node (16_14) at (4.000, 5.600) {$\bullet$};
            \node (p) at (0.000, 5.200) {$p$};
            \node (2_7) at (0.500, 2.800) {$\bullet$};
            \node (r2) at (1.000, 0.800) {$r_2$};
            \node (7_14) at (1.750, 5.600) {$\bullet$};
            \node (8_13) at (2.000, 5.200) {$\bullet$};
            \node (0_11) at (0.000, 4.400) {$\bullet$};
            \node (r1) at (4.500, 7.600) {$r_1$};
            \node (7_6) at (1.750, 2.400) {$\bullet$};
            \node (q) at (4.500, 2.000) {$q$};
            \node (1_16) at (0.250, 6.400) {$\bullet$};
            \node (7_18) at (1.750, 7.200) {$\bullet$};
            \node (17_7) at (4.250, 2.800) {$\bullet$};
            \node (7_4) at (1.750, 1.600) {$\bullet$};
            \node (5_8) at (1.250, 3.200) {$\bullet$};

            \graph {(0_11) ->[red] (p)};
            \graph {(p) ->[red] (1_16)};
            \graph {(1_16) ->[red] (7_18)};
            \graph {(7_18) ->[red] (r1)};
            \graph {(r1) ->[red] (q)};
            \graph {(q) ->[red] (6_2)};
            \graph {(6_2) ->[red] (r2)};
            \graph {(r2) ->[red] (0_11)};
            \graph {(p) --[black] (q)};
            \graph {(p) --[black] (r1)};
            \graph {(r1) --[black] (q)};
            \graph {(q) --[black] (r2)};
            \graph {(r2) --[black] (p)};
        \end{tikzpicture}
    }
    \caption{The points $p$ and $q$ are the left-most and 
             right-most points, and therefore in the hull. The points $r_1$,
             $r_2$ are farthest from the line $pq$, and also in the hull. Any
             point within the two triangles $\Delta pr_1 q$, $\Delta q r_2 p$
             cannot be in the convex hull.}
    \label{fig:quickhull_corr}
\end{figure}

\begin{algorithm}[ht]
\begin{algorithmic}[1]
    \Require $P: $ set of points in $\mathbb{R}^2$.
    \Ensure $H: $ vertices of convex hull of $P$ in clockwise order.
    \State Let $p$ be the left-most point
    \State Let $q$ the right-most point
    \State $S_1 := \{u \in P \mid u \text{ to the left of }pq\}$
    \State $r_1 := \argmax\limits_{u \in S_1}d(u, pq)$.
    \State $S_2 := \{u \in P \mid u \text{ to the right of }pq\}$
    \State $r_2 := \argmax\limits_{u \in S_2}d(u, pq)$.
    \State $H := \{p\} \cup $ \Call{Hull}{$S_1$, $p$, $r_1$, $q$}
            $\cup \{q\} \cup$ \Call{Hull}{$S_2$, $q$, $r_2$, $p$}.
    \Function{Hull}{P, p, r, q}
        \If{$|P| \leq 1$}
            \State $H := P$.
        \Else{}
            \State $S_1 := \{u \in P \mid u \text{ to the left of }pr\}$
            \State $r_1 := \argmax\limits_{u \in S_1}d(u, pr)$.
            \State $S_2 := \{u \in P \mid u \text{ to the left of }rq\}$
            \State $r_2 := \argmax\limits_{u \in S_2}d(u, rq)$.
            \State $H := \{p\} \cup$ \Call{Hull}{$S_1$, $p$, $r_1$, $r$}
                   $\cup \{r\} \cup$ \Call{Hull}{$S_2$, $r$, $r_2$, $q$} 
                   $\cup \{q\}$.
        \EndIf
    \EndFunction
\end{algorithmic}
\captionof{algorithm}{Quickhull algorithm.}
\label{alg:quickhull_corr}
\end{algorithm}

While it is intuitively clear what points are to the left of a line $pq$, or
inside a triangle $prq$, we need a quantity called the orientation of three 
points to define this precisely.
Given three points $p, u, q$ in the plane, we have vectors
\[
    \vec{up} = \begin{pmatrix}
        p_x - u_x \\
        p_y - u_y
    \end{pmatrix},
    \vec{uq} = \begin{pmatrix}
        q_x - u_x \\
        q_y - u_y
    \end{pmatrix}
\]
representing oriented line segments between the points, as depicted in
Figure~\ref{fig:orient1}.
\begin{figure}[ht]
    \begin{subfigure}{0.45\textwidth}
    \begin{tikzpicture}
        \node (p) at (0, 0) {$p$};
        \node (u) at (1, 2) {$u$};
        \node (q) at (2, 1) {$q$};

        \graph {(u) ->["$\vec{up}$"'] (p)};
        \graph {(u) ->["$\vec{uq}$" ] (q)};
    \end{tikzpicture}
    \subcaption{Two vectors induced by three points in the plane. The path 
             $p, u, q$ makes a right turn.}
    \label{fig:orient1}
    \end{subfigure}\hfill
    \begin{subfigure}{0.45\textwidth}
    \begin{tikzpicture}
        \node (p) at (0, 0)  {$p$};
        \node (u) at (1, 2)  {$u$};
        \node (q) at (2, 1)  {$q$};
        \node (c) at (1, -1) {};

        \graph {(u) ->["$\vec{up}$"'] (p)};
        \graph {(u) ->["$\vec{uq}$" ] (q)};
        \graph {(p) --                (c)};
        \graph {(c) --                (q)};
    \end{tikzpicture}
    \subcaption{Parallelogram spanned by $\vec{up}$ and $\vec{uq}$. 
             The area of this parallelogram is $|\orient(p, u, q)|$.}
    \label{fig:orient2}
    \end{subfigure}
    \caption{Vectors between points in the plane.}
\end{figure}
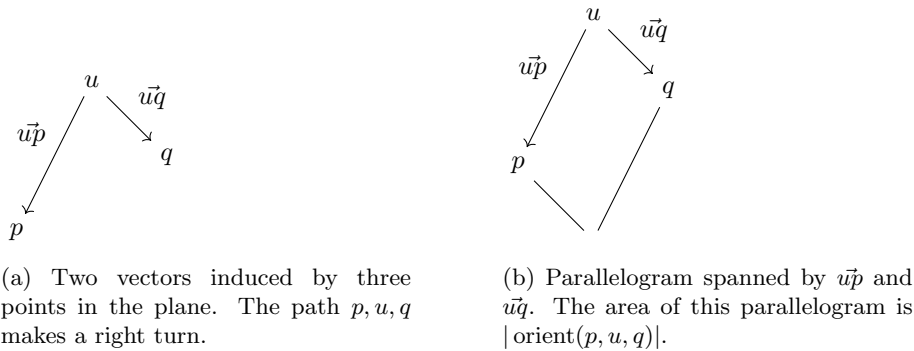

We can view these as vectors in $\mathbb{R}^3$ by adding a $z$-coordinate $0$. 
Now the cross product $\vec{up} \times \vec{uq}$ is equal to 

\[
    \begin{pmatrix}
        0 \\
        0 \\
        (p_x - u_x) \cdot (q_y - u_y) - (p_y - u_y) \cdot (q_x - u_x)
    \end{pmatrix}.
\]

We write $\orient(p, u, q)$ for the $z$-coordinate
$(p_x - u_x) \cdot (q_y - u_y) - (p_y - u_y) \cdot (q_x - u_x)$
If the cross-product points up, so if $\orient(p, u, q) > 0$, we make a 
right turn going from $p$ to $u$ to $q$, and if this is negative, we make a 
left turn. If it is zero, the points $p, u, q$ are collinear. 
This gives us the precise definition of what it means for $u$ to be to the left
of $pq$: $\orient(p, u, q) > 0$.

% p = (0, 0), u = (0, 1), q = (1, 1) is a right-hand turn.
% (p_x - u_x) \cdot (q_y - u_y) - (p_y - u_y) \cdot (q_x - u_x) =
% (0 - 0) * (1 - 1) - (0 - 1) * (1 - 0) = 1

We can also use the orientation to find the point farthest from the line segment 
$pq$~\cite{quickerthanqhull}. 
The norm of the outer product
$|\vec{up} \times \vec{uq}| = |\orient(p, u, q)|$ is the area of the
parallelogram spanned by $\vec{up}$ and $\vec{uq}$. This is also equal to
the length of $pq$ multiplied by the distance from $u$ to the line through 
$pq$. So $|\orient(p, u, q)| > |\orient(p, u', q)|$ if and only if $u$ is 
farther from $pq$ than $u'$.

\subsection{Numerical Robustness}\label{subsec:robustness}

It is unfeasible to compute exact solutions to real-valued functions because
the input is imprecise. This imprecision can arise from the rounding of values
in $\mathbb{R}$ to a finite data type \texttt{double} or \texttt{float}, but
also from imprecision in measuring physical quantities, or truncation errors
in approximating continuous functions with discrete versions.

Instead of aiming to compute the correct result, we aim to find bounds on
the error.
The error between the true solution $f(x)$ of a problem $x$, and the computed 
solution $\widehat{f(x)}$ is called the \textit{forward error}. 
We may also ask ourselves: for what $\tilde{x}$ is our computed solution 
$\widehat{f(x)}$ the actual solution $f(\tilde{x})$? We call the error between 
$x$ and $\tilde{x}$ the \textit{backward error}. If the backward
error is smaller than the uncertainty in the input, the solution is as good as 
we can expect.

We can connect these two ideas together with a
\textit{perturbation analysis}:
how is the error between $x$ and $\tilde{x}$ in the input magnified to an error
in the output between $f(x)$ and $f(\tilde{x})$?
If the error in the output has the same or smaller magnitude we call the problem 
\textit{well-conditioned}, if the error in the output is substantially 
larger we call the problem \textit{ill-conditioned}. For a well-conditioned 
problem, a small backward error also implies a small forward error.

Algorithms whose forward or backward error grows slowly with respect to the
problem size are called \textit{numerically stable}.

\begin{example}\label{ex:cancellation}
    Subtraction is a classic example, and will also be useful for our analysis.
    For non-zero numbers in $\mathbb{R}$, we define the relative error of 
    $\hat{x}$ approximating $x$ as the quantity $|x - \hat{x}| / |x|$, or 
    equivalently $|\delta|$ for $\hat{x} = (1 + \delta)x$. The computed 
    difference $fl(a - b)$ of floating point numbers $a$, $b$ satisfies 
    $fl(a - b) = (1 + \delta)(a - b)$, where $\delta$ is smaller than the
    machine precision $\epsilon$, $2^{-24}$ for single-precision and 
    $2^{-53}$ for double-precision. This means we have a forward error bounded 
    by $\epsilon$, and by rewriting to 
    $fl(a - b) = (1 + \epsilon)a - (1 + \epsilon)b$, also the 
    backward error is bounded by $\epsilon$. For the perturbation analysis, we 
    subtract two perturbed values $\hat{a} = (1 + \delta_a)a$ and 
    $\hat{b} = (1 + \delta_b)b$ and compute the relative error
    $$\left|\frac{(a - b) - (\hat{a} - \hat{b})}{(a - b)}\right| = 
    \frac{|\delta_a a + \delta_b b|}{|a - b|} \leq 
    \max(|\delta_a|, |\delta_b|) \frac{|a| + |b|}{|a - b|}.$$
    
    We see that subtraction is ill-conditioned when $a \approx b$: small errors 
    $\delta_a$, $\delta_b$ in the input are magnified by the subtraction. The 
    reader may recognize this as catastrophic cancellation. Note that this is 
    only problematic if $a$, $b$ have been contaminated by error.
\end{example}

\section{Perturbation Analysis}

To bound an error, we first need to define how far a computed solution is
from the true solution. As we work with geometric shapes, it is intuitive
to define error in terms of distance. The standard metric is the Hausdorff 
distance 
$d(A, B) := \max\left(\max_{a \in A}d(a, B), \max_{b \in B}d(A, b)\right)$. 
This has as unfortunate property that the distance is not invariant to scaling:
we have $d(\{(1.00, 1.00)\}, \{(2.00, 2.00)\}) = \sqrt{2}$, but 
$d(\{(100, 100)\}, \{(200, 200)\}) = 100 \sqrt{2}$, while these may very well
be the same data, one expressed in metres, and the other in centimetres.
If we have a reference set---in our case the input---we can divide the Hausdorff
distance by the maximum absolute value of $x$- and $y$-coordinates $M$, 
yielding $d_M(A, B) := d(A, B) / M$. This makes the measure of error independent
of the chosen unit. The resulting function $d_M$ satisfies the axioms of a 
distance.

We can now do the perturbation analysis with respect to this error. We write
$\delta = d_M(\tilde{P}, P)$ for the error between $P$ and perturbed set 
$\tilde{P}$. To bound the Hausdorff distance we take a point $a \in CH(P)$. 
We can write this as a convex combination 
$a = \sum_{i = 1}^n \lambda_i a_i$ for $a_i \in P$
and $\sum_{i = 1}^n \lambda_i = 1$, $\lambda_i \geq 0$. 
By definition of the Hausdorff distance, there exist $b_i \in \tilde{P}$
such that $d(a_i, b_i) \leq M \delta$. This gives a point 
$b = \sum_{i = 1}^n \lambda_i b_i$ that is on the convex hull of $\tilde{P}$,
and close to $a$. The triangle inequality bounds the distance
\[
d(a, b) = |a - b| \leq \sum_{i = 1}^n \lambda_i |a_i - b_i| 
\leq \sum_{i = 1}^n \lambda_i M\delta = \delta M.
\]
Therefore $d(CH(P), CH(\tilde{P})) \leq \delta M$, so 
$d_M(CH(P), CH(\tilde{P})) \leq \delta = d_M(P, \tilde{P})$. 
As the backward error 
leads to a forward error of the same magnitude, we can consider the convex hull
problem well-conditioned.

\section{Geometric Tests}

As preparation for the full analysis, we first consider the steps of 
Algorithm~\ref{alg:quickhull_corr} in isolation.
Recall from the background section that we can decide whether $puq$ makes a 
right turn by testing
$\orient(p, u, q) > 0$. To reduce the number of floating point operations 
necessary to evaluate this in a loop over $u$, we rewrite the inequality
to $(p_x - u_x) \cdot (q_y - p_y) > (p_y - u_y) \cdot (q_x - p_x)$. The
subtractions $q_y - p_y$ and $q_x - p_x$ can be computed once and reused every
iteration. We denote the true evaluation of this inequality by $\rt(p, u, q)$ 
and the computed evaluation by $\widehat{\rt}(p, u, q)$.

We can test whether $u$ is farther from $pq$ than $u'$ by comparing
$\orient(p, u, q)$ to $\orient(p, u', q)$. This makes it tempting
to compute the orientation once for each point, and then do both
the right-turn test and the distance test with this quantity. However, this
can run into precision problems because the middle subtraction in
$(u_x - p_x)(q_y - u_y) \bm{-} (u_y - p_y)(q_x - u_x)$
takes arguments that are polluted by error and this might incur
numerical cancellation (Example~\ref{ex:cancellation}).
Instead, we rewrite the inequality
\begin{align*}
    (p_x - u_x)(q_y - u_y) &\bm{-} (p_y - u_y)(q_x - u_x) < \\
    (p'_x - u_x)(q_y - u'_y) &\bm{-} (p'_y - u_y)(q_x - u'_x)
\end{align*}
to
\begin{align*}
(q_y - p_y) (u_x - u'_x) < (q_x - p_x) (u_y - u'_y).
\end{align*}
This eliminates any subtractions on computed values.
We write $\frt(p, q, u, u')$ for the predicate $d(pq, u) > d(pq, u')$,
and $\widehat{\frt}(p, q, u, u')$ for the computed predicate.

Finding the point farthest from $pq$ is a reduction
\[
    u_1 \op \cdots \op u_2
\]
where $u_1 \op u_2$ returns $u_1$ if $\frt(p, q, u_1, u_2)$ and $u_2$ 
otherwise. Mathematically this gives a point with maximum distance to $pq$
regardless of how we order the operations. However, this is not true for the 
floating point evaluation $\widehat{\frt}$. We write this reduction as
$\reduce \widehat{\frt}(p, q, -, -) X$ for some choice evaluation order
over points in $X$. We will analyse some choices in 
Subsection~\ref{subsec:farther}.

Taken together, this gives Algorithm~\ref{alg:quickhull_float}.

\begin{algorithm}[ht]
\begin{algorithmic}[1]
    \Require $P: $ set of points in $\mathbb{R}^2$.
    \Ensure $H: $ vertices of convex hull of $P$ in clockwise order.
    \State Let $p$ be the left-most point
    \State Let $q$ the right-most point
    \State $S_1 := \{u \in P \mid \widehat{rt}(p, u, q)\}$.
    \State $r_1 := \reduce \widehat{\frt}(p, q, -, -) S_1$.
    \State $S_2 := \{u \in P \mid \widehat{rt}(q, u, p)\}$.
    \State $r_2 := \reduce \widehat{\frt}(q, p, -, -) S_2$.
    \State $H := \{p\} \cup $ \Call{Hull}{$S_1$, $p$, $r_1$, $q$}
            $\cup \{q\} \cup$ \Call{Hull}{$S_2$, $q$, $r_2$, $p$}.
    \Function{Hull}{$P, p, r, q$}
        \If{$|P| \leq 1$}
            \State $H := P$.
        \Else{}
            \State $S_1 := \{u \in P \mid \widehat{\rt}(p, u, r)\}$.
            \State $r_1 := \reduce \widehat{\frt}(p, r, -, -) S_1$.
            \State $S_2 := \{u \in P \mid \widehat{\rt}(r, u, q)\}$
            \State $r_2 := \reduce \widehat{\frt}(r, q, -, -) S_2$.
            \State $H := \{p\} \cup$ \Call{Hull}{$S_1$, $p$, $r_1$, $r$}
                   $\cup \{r\} \cup$ \Call{Hull}{$S_2$, $r$, $r_2$, $q$} 
                   $\cup \{q\}$.
        \EndIf
    \EndFunction
\end{algorithmic}
\captionof{algorithm}{Floating point Quickhull algorithm}
\label{alg:quickhull_float}
\end{algorithm}

\subsection{Right Turn}

Lemma~\ref{lem:right-turn} shows $\widehat{rt}(p, u, q)$ is accurate for points 
far from $pq$. This is a desirable property because we only misclassify points
that have little effect on result.

\begin{lemma}\label{lem:right-turn}
    If $\rt(p, u, q)$ is true, but $\widehat{\rt}(p, u, q)$ is not, or the
    other way around, then $d(u, pq) \leq \gamma_6 M$, where 
    $\gamma_6 = \frac{6\epsilon}{1 - 6\epsilon}$ for $\epsilon$ the machine 
    precision.
\end{lemma}

\begin{proof}
    As $\orient(p, u, q)$ is equal to the parallelogram spanned by
    $\vec{up}$ and $\vec{uq}$, it is equal to twice the area of triangle
    $puq$. This gives an upper bound
    \begin{align*}
        d(u, pq) \leq \frac{|\orient(p, u, q)|}{2\lVert p - q \rVert},
    \end{align*}
    so we start by showing that a conflict between $\widehat{\rt}$ and $\rt$ 
    implies small $|\orient(p, u, q)|$.

    Suppose $\rt(p, u, q)$ is true, but $\widehat{\rt}(p, u, q)$ is false.
    We write $fl(\cdots)$ for the floating point evaluation of a primitive 
    operation. As the IEEE-754 standard satisfies
    $fl(x \circ y) = (1 + \delta)(x \circ y)$, $|\delta| < \epsilon$ for
    $\circ \in \{+, -, \cdot\}$, the test
    \begin{align*}
        \widehat{\rt}(p, u, q) \equiv 
        fl(fl(p_x - u_x) fl(q_y - p_y)) < fl(fl(p_y - u_y) fl(q_x - p_x))
    \end{align*}
    implies 
    \begin{align}\label{eq:rt:1}
    (p_x - u_x) (q_y - p_y) <& \frac{(1 + \delta_1)(1 + \delta_2)(1 + \delta_3)}
    {(1 + \delta_4)(1 + \delta_5)(1 + \delta_6)} (p_y - u_y) (q_x - p_x).
    \end{align}
    We can simplify this inequality by noting there exists a $\theta$ such that
    \begin{align*}
    \frac{(1 + \delta_1)(1 + \delta_2)(1 + \delta_3)}{(1 + \delta_4)(1 + \delta_5)(1 + \delta_6)} = 1 + \theta
    \end{align*}
    and $|\theta| \leq \gamma_6$ \cite[Lemma~3.1]{Higham02}.
    That yields
    \begin{align*}
        (p_x - u_x) (q_y - p_y) &< (1 + \theta) (p_y - u_y) (q_x - p_x) \equiv \\
        -\theta (p_y - u_y) (q_x - p_x) &< \orient(p, u, q).
    \end{align*}
    As $\rt(p, u, q)$ is false, we have $\orient(p, u, q) \leq 0$, so
    $|\orient(p, u, q)| \leq \gamma_6 |(p_y - u_y) (q_x - p_x)|$. An analogous
    argument shows this also holds for $\rt(p, u, q)$ true but 
    $\hat{\rt}(p, u, q)$ false.

    We can now compute an upper bound
    \begin{align*}
        d(u, pq) &= \frac{|\orient(p, u, q)|}{2 \lVert p - q \rVert} \\
        &\leq \frac{\gamma_6 |(p_y - u_y) (q_x - p_x)|}{2 \sqrt{(q_x - p_x)^2 + (q_y - p_y)^2}} \\
        &\leq \frac{1}{2}\gamma_6 |p_y - u_y| \\
        &\leq \frac{1}{2} \gamma_6 (|p_y| + |u_y|) \\
        &\leq \gamma_6 M.
    \end{align*}
\end{proof}

\subsection{Distance Test}\label{subsec:farther}

Lemma~\ref{lem:farther} shows that at least the building blocks of the 
reductions---individual tests $\widehat{\frt}$---are accurate.

\begin{lemma}\label{lem:farther}
    If $\widehat{\frt}(p, q, u, u')$ is true, but $\frt(p, q, u, u')$ is false,
    then $0 \leq d(u', pq) - d(u, pq) \leq \gamma_6 M$.
\end{lemma}

\begin{proof}
    Similarly to Lemma~\ref{lem:right-turn} we assume 
    $\widehat{\frt}(p, q, u, u')$ true, but $\frt(p, q, u, u')$ false.
    This gives the two inequalities
    \begin{align*}
        (q_y - p_y) (u_x - u'_x) - (q_x - p_x) (u_y - u'_y) &< 
            \theta(q_x - p_x) (u_y - u'_y), \\
        (q_y - p_y) (u_x - u'_x) - (q_x - p_x) (u_y - u'_y) &\geq 0,
    \end{align*}
    for some $|\theta| \leq \gamma_6$, which we can combine to 
    \begin{align}\label{eq:frt1}
        |(q_y - p_y) (u_x - u'_x) - (q_x - p_x) (u_y - u'_y)| \leq 
            |\theta(q_x - p_x) (u_y - u'_y)|.
    \end{align}

    Looking at the left side from a geometric point of view gives

    \begin{equation}\label{eq:frt2}
    \begin{aligned}
        (q_y - p_y) (u_x - u'_x) - (q_x - p_x) (u_y - u'_y) &
        = \orient(p, u, q) - \orient(p, u', q) \\
        & = 2 d(u, pq) \lVert p - q \rVert - 2 d(u', pq) \lVert p - q \rVert \\
        & = 2 \lVert p - q \rVert (d(u, pq) - d(u', pq)).
    \end{aligned}
    \end{equation}

    Combining the bound of Equation~\ref{eq:frt1} with the geometric view of
    Equation~\ref{eq:frt2} gives
    \begin{align*}
        d(u, pq) - d(u', pq) & \leq
       |\theta|\frac{|(q_x - p_x) (u_y - u'_y)|}{2 \lVert p - q \rVert} \\
        & \leq \frac{1}{2}|\theta| |u_y - u'_y| \leq |\theta|M.
    \end{align*}
\end{proof}

The accuracy combining these $\widehat{\frt}$ tests depends on the evaluation
order. We analyse three orders inspired by summation. The three
algorithms~\ref{alg:farthest1},~\ref{alg:farthest2},~\ref{alg:farthest3}
resemble normal summation, blocked summation, and pairwise summation
respectively. Lemma~\ref{lem:farther_reduc} gives bounds for these algorithms.

\begin{algorithm}[ht]
\begin{algorithmic}[1]
    \Require $P: $ ordered set of points in $\mathbb{R}^2$, 
            $p, q \in \mathbb{R}^2$.
    \Ensure $u_{\max} \in P$ farthest from $pq$ with positive 
            $\orient(p, r, q)$.
    \State $u_{\max} = P[0]$;
    \For{$i = 1$ \textbf{to} $|P| - 1$}
        \If{$\widehat{\frt}(p, q, P[i], u_{\max})$}
            \State $u_{\max} = P[i]$;
        \EndIf
    \EndFor
\end{algorithmic}
    \captionof{algorithm}{Simple way of finding farthest point from $pq$.}
\label{alg:farthest1}
\end{algorithm}

\begin{algorithm}[ht]
\begin{algorithmic}[1]
    \Require $P: $ ordered set of points in $\mathbb{R}^2$, 
            $p, q \in \mathbb{R}^2$.
    \Ensure $u_{\max} \in P$ farthest from $pq$ with positive 
            $\orient(p, r, q)$.
    \State $u_{\max} = P[0]$;
    \For{$I = 1$ \textbf{to} $|P| - 1$ \textbf{step} $m$}
        \State $b_{\max} = P[I]$;
        \For{$i = I + 1$ \textbf{to} $\min(I + m - 1, |P| - 1)$}
            \If{$\widehat{\frt}(p, q, P[i], b_{\max})$}
                \State $b_{\max} = P[i]$;
            \EndIf
        \EndFor
        \If{$\widehat{\frt}(p, q, b_{\max}, u_{\max})$}
            \State $u_{\max} = b_{\max}$
        \EndIf
    \EndFor
\end{algorithmic}
    \captionof{algorithm}{Finding farthest point from $pq$ with blocking.}
\label{alg:farthest2}
\end{algorithm}

\begin{algorithm}[ht]
\begin{algorithmic}[1]
    \Require $P: $ ordered set of points in $\mathbb{R}^2$, 
            $p, q \in \mathbb{R}^2$.
    \Ensure $u_{\max} \in P$ farthest from $pq$ with positive 
            $\orient(p, r, q)$.
    \Function{Farthest}{P, p, q}
        \If{$|P| = 1$}
            \State $u_{\max} = P[0]$.
        \Else{}
            \State $u^{1}_{\max} = \Call{Farthest}{P(1 : |P| / 2 - 1), p, q}$.
            \State $u^{2}_{\max} = \Call{Farthest}{P(|P| / 2 : |P| - 1), p, q}$.
            \If{$\widehat{\frt}(p, q, u^{1}_{\max}, u^{2}_{\max})$}
                \State $u_{\max} = u^{1}_{\max}$.
            \Else{}
                \State $u_{\max} = u^{2}_{\max}$.
            \EndIf
        \EndIf
    \EndFunction
\end{algorithmic}
\captionof{algorithm}{Finding farthest point from $pq$ with recursion.}
\label{alg:farthest3}
\end{algorithm}

\begin{lemma}\label{lem:farther_reduc}
    Let $r$ be the farthest point from $pq$ out of $n$ points and 
    $\widehat{r}$ a computed estimate of this point.
    We write 
    \[
        F_n := \frac{d(r, pq) - d(\widehat{r}, pq)}{M}.
    \]
    for the scaled error. 
    Algorithms~\ref{alg:farthest1},~\ref{alg:farthest2},~\ref{alg:farthest3}
    have asymptotic error bounds $O(n\epsilon)$, $O((m + n / m) \epsilon)$,
    $O(\log(n)\epsilon)$ on $F_n$ respectively.
\end{lemma}

\begin{proof}
    We first show how the error propagates through chains of incorrect 
    decisions. We then prove the three cases by bounding the length of such
    a chain. Suppose $\widehat{\frt}(p, q, u, u')$ is evaluated correctly.
    Then $d(u, pq) - d(u', pq) > 0$, so 
    $d(r, pq) - d(u, pq) < d(r, pq) - d(u', pq)$, implying $u$ is a better guess 
    than $u'$. If $\widehat{\frt}(p, q, u, u')$ is evaluated incorrectly, then
    Lemma~\ref{lem:farther} shows $d(u', pq) - d(u, pq) < \gamma_6 M$, or
    $d(r, pq) - d(u, pq) < d(r, pq) - d(u', pq) + \gamma_6 M$. So $u'$ makes
    the guess at most $\gamma_6 M$ worse.

    The first bound follows directly from doing at most $n$ comparisons and
    $\gamma_6 \in O(\epsilon)$. If we do this within a group of size $m$, we may
    end up with an error of $O(m \epsilon)$. We then compare the farthest point
    within each of these $n / m$ groups, so the final error may be 
    $O((m + n / m) \epsilon)$.
    The final bound follows from the recursion having depth 
    $\lceil \log_2 n \rceil$ and an induction.
\end{proof}

Example~\ref{counterexample} shows we can get close to the worst-case bound 
$O(n \epsilon)$ even when assuming the rounding is random, as long as we can 
construct an adversarial input.

\begin{example}\label{counterexample}
    Consider the points in Figure~\ref{fig:counter} where 
    $d(u_{i}, pq) - d(u_{i + 1}, pq) \approx \gamma_6 M / k$ for some
    fixed $k$.
    We will keep updating our guess of the farthest point unless we
    evaluate $\widehat{\frt}$ correctly $k$ times in a row. The 
    probability of this happening is exponential in $k$, so we do not have to 
    choose $k$ large relative to $n$ to likely end up with a point $u_i$ close 
    to $u_n$.
    This point satisfies
    \[
        \frac{d(u_1, pq) - d(u_i, pq)}{M} \in O(n\epsilon).
    \]

    \begin{figure}
    \resizebox{0.5 \columnwidth}{!}{%
        \begin{tikzpicture}
            \node (p) at (0, 0) {$p$};
            \node (q) at (5, 5) {$q$};
            \foreach \i in {1, ..., 10} {
                \node at (0 - \i / 6, 0 + \i / 5) {$.$};
            }
            \foreach \i in {1, ..., 10} {
                \node at (5 - \i / 6, 5 + \i / 5) {$.$};
            }
            \node at (0 - 11 / 6, 0 + 11 / 5) {$u_1$};
            \node at (0 - 5 / 6, 0 + 5 / 5) {$u_{n / 2}$};
            \node at (5 - 11 / 6, 5 + 11 / 5) {$u_2$};
            \node at (5 - 5 / 6, 5 + 5 / 5) {$u_{n / 2 + 1}$};
            \graph {(p) -- (q)};
        \end{tikzpicture}
    }
    \caption{Adversarial input}\label{fig:counter}
    \end{figure}
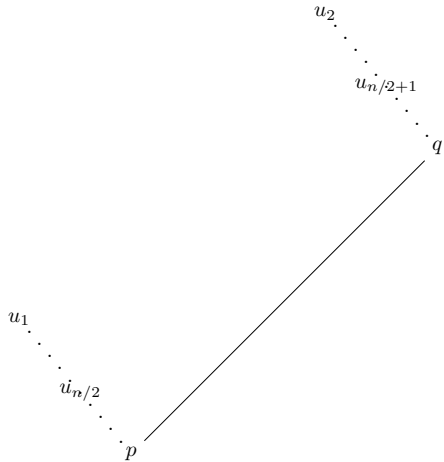
\end{example}

Though the structure of the points is not especially adversarial, the
order of points is. We hypothesize that even Algorithm~\ref{alg:farthest1}
will be close to optimal, if the order of points is random.

\begin{hyp}\label{hyp:farther_reduc_random}
    Under random permutations, Algorithm~\ref{alg:farthest1} gives
    $\mathbb{E}[F_{|P|}] \in O(\epsilon)$.
\end{hyp}

A further investigation of $F_{|P|}$ will show this hypothesis is related
to the expected length of subsequences, making it a more complicated variation 
of Ulam's Problem~\cite{Romik15} and therefore outside the scope of this paper.
Instead, we will support Hypothesis~\ref{hyp:farther_reduc_random} with a 
Monte Carlo method.

Let $u_1$, $\cdots$, $u_{|P|}$ be ordered by distance from $pq$.
We write $u \sim u'$ for the relation $|d(u, pq) - d(u', pq)| / M < \gamma_6$,
e.g. we cannot tell at our precision whether $u$ or $u'$ is farther from $pq$.
If we make the wrong decision every time we evaluate $\widehat{frt}(u, u', p, q)$,
a permutation $\sigma$ will end up with $u_{\sigma(|P| - m)}$ for the largest 
subsequence $u_{\sigma(|P|)} \sim u_{\sigma(|P| - 1)} \sim \cdots 
\sim u_{\sigma(|P| - m)}$. We have an upper bound of $m\gamma_6$ on $F_{|P|}$. 

The more points that are pairwise related, the larger the expected length of
such a subsequence. However, for such a subsequence $u_{i + k - 1}, \cdots, u_i$
satisfying $u_{i + k - 1} \sim u_i$, each wrong guess in a subsequence only 
increases $F_n$ by $\gamma_6 / k$ on average. To capture this effect, we create 
subsequences for different values of $k$, similar to the construction of 
Example~\ref{counterexample}.
We let $|P|$ range from $256$ to $2^{20}$, $k$ from $1$ to $10$, take $300$ 
samples and repeat each experiment $10$ times, reporting
error bars in Figure~\ref{fig:monte_carlo}.
We see that $\mathbb{E}[F_{|P|}]$ appears to be constant in $k$ and
$|P|$. The code for reproducing this experiment is available 
at~\cite{quickhull-correctness-artefact}.

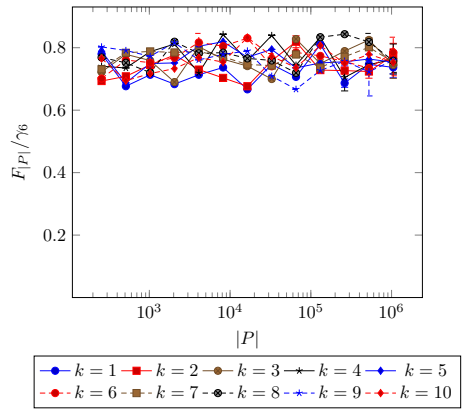
\begin{figure}[ht]
    \begin{adjustbox}{width=0.5\textwidth}
        \begin{tikzpicture}
            \begin{semilogxaxis}[
                xlabel=$|P|$,
                ylabel=$F_{|P|} / \gamma_6$,
                legend style={at={(0.5,-0.2)},anchor=north},
                legend columns = 5,
                xticklabel style={/pgf/number format/fixed},
                yticklabel style={/pgf/number format/fixed},
                ymin=0,
                ytick={0.2, 0.4, 0.6, 0.8, 1.0},
                yticklabels={0.2, 0.4, 0.6, 0.8, 1.0},
            ]

%                \foreach \k in {1, ..., 10} {
%                    \addplot+ [
%                        error bars/.cd,
%                            y explicit,
%                            y dir=both,
%                        ] table [
%                            x = n,
%                            y = mean,
%                            y error plus = stddev,
%                            y error minus = stddev,
%                        ] {monte_carlo_\k.csv};
%                }

                \addplot+ [
                    error bars/.cd,
                        y explicit,
                        y dir=both,
                    ] table [
                        x = n,
                        y = mean,
                        y error plus = stddev,
                        y error minus = stddev,
                    ] {
                            n, mean, stddev
                            256, 0.780000,0.000000
                            512, 0.676667,0.000000
                            1024, 0.713333,0.000000
                            2048, 0.683333,0.000000
                            4096, 0.713333,0.000000
                            8192, 0.736667,0.000000
                            16384, 0.666667,0.000000
                            32768, 0.746667,0.000000
                            65536, 0.706667,0.000000
                            131072, 0.810000,0.000000
                            262144, 0.686667,0.013333
                            524288, 0.746667,0.032660
                            1048576, 0.737000,0.033248
                    };

                \addplot+ [
                    error bars/.cd,
                        y explicit,
                        y dir=both,
                    ] table [
                        x = n,
                        y = mean,
                        y error plus = stddev,
                        y error minus = stddev,
                    ] {
                        n, mean, stddev
                        256, 0.693333,0.000000
                        512, 0.708333,0.000000
                        1024, 0.745000,0.000000
                        2048, 0.776667,0.000000
                        4096, 0.730000,0.000000
                        8192, 0.703333,0.000000
                        16384, 0.676667,0.000000
                        32768, 0.758333,0.000000
                        65536, 0.819167,0.019094
                        131072, 0.728333,0.000000
                        262144, 0.726667,0.000000
                        524288, 0.723833,0.021422
                        1048576, 0.775333,0.038700
                    };

                \addplot+ [
                    error bars/.cd,
                        y explicit,
                        y dir=both,
                    ] table [
                        x = n,
                        y = mean,
                        y error plus = stddev,
                        y error minus = stddev,
                    ] {
n, mean, stddev
256, 0.722222,0.000000
512, 0.778885,0.000000
1024, 0.758890,0.000000
2048, 0.689983,0.000000
4096, 0.803335,0.000000
8192, 0.763272,0.000000
16384, 0.741110,0.000000
32768, 0.699733,0.000000
65536, 0.827800,0.000000
131072, 0.749043,0.012246
262144, 0.789167,0.000000
524288, 0.824833,0.006688
1048576, 0.744667,0.042838
};
                \addplot+ [
                    error bars/.cd,
                        y explicit,
                        y dir=both,
                    ] table [
                        x = n,
                        y = mean,
                        y error plus = stddev,
                        y error minus = stddev,
                    ] {
n, mean, stddev
256, 0.740833,0.000000
512, 0.735000,0.000000
1024, 0.787500,0.000000
2048, 0.811667,0.000000
4096, 0.719167,0.000000
8192, 0.842500,0.000000
16384, 0.758333,0.000000
32768, 0.839167,0.000000
65536, 0.744167,0.000000
131072, 0.833333,0.000000
262144, 0.705000,0.043333
524288, 0.745500,0.012656
1048576, 0.763750,0.048187
};

                \addplot+ [
                    error bars/.cd,
                        y explicit,
                        y dir=both,
                    ] table [
                        x = n,
                        y = mean,
                        y error plus = stddev,
                        y error minus = stddev,
                    ] {
n, mean, stddev
256, 0.788667,0.000000
512, 0.687331,0.000000
1024, 0.749340,0.000000
2048, 0.751327,0.000000
4096, 0.804667,0.000000
8192, 0.819291,0.000000
16384, 0.769433,0.000000
32768, 0.794577,0.000000
65536, 0.735368,0.000000
131072, 0.751315,0.000000
262144, 0.754893,0.013163
524288, 0.763760,0.014620
1048576, 0.753562,0.036820
};

                \addplot+ [
                    error bars/.cd,
                        y explicit,
                        y dir=both,
                    ] table [
                        x = n,
                        y = mean,
                        y error plus = stddev,
                        y error minus = stddev,
                    ] {
n, mean, stddev
256, 0.701111,0.000000
512, 0.760553,0.000000
1024, 0.752223,0.000000
2048, 0.767769,0.000000
4096, 0.817002,0.028849
8192, 0.804963,0.000000
16384, 0.830560,0.000000
32768, 0.740418,0.000000
65536, 0.785013,0.000000
131072, 0.773854,0.000000
262144, 0.755656,0.006937
524288, 0.734555,0.010585
1048576, 0.784760,0.048904
};

                \addplot+ [
                    error bars/.cd,
                        y explicit,
                        y dir=both,
                    ] table [
                        x = n,
                        y = mean,
                        y error plus = stddev,
                        y error minus = stddev,
                    ] {
n, mean, stddev
256, 0.730951,0.000000
512, 0.783810,0.000000
1024, 0.788092,0.000000
2048, 0.784751,0.000000
4096, 0.787907,0.029394
8192, 0.770919,0.000000
16384, 0.746102,0.000000
32768, 0.742869,0.000000
65536, 0.778311,0.000000
131072, 0.736413,0.000000
262144, 0.771026,0.016194
524288, 0.801281,0.009734
1048576, 0.753911,0.035809
};

                \addplot+ [
                    error bars/.cd,
                        y explicit,
                        y dir=both,
                    ] table [
                        x = n,
                        y = mean,
                        y error plus = stddev,
                        y error minus = stddev,
                    ] {
n, mean, stddev
256, 0.768750,0.000000
512, 0.751667,0.000000
1024, 0.720417,0.000000
2048, 0.819167,0.000000
4096, 0.782083,0.000000
8192, 0.784167,0.000000
16384, 0.765833,0.000000
32768, 0.759167,0.000000
65536, 0.719583,0.000000
131072, 0.833750,0.000000
262144, 0.843333,0.000000
524288, 0.818375,0.027459
1048576, 0.757792,0.036688
};

                \addplot+ [
                    error bars/.cd,
                        y explicit,
                        y dir=both,
                    ] table [
                        x = n,
                        y = mean,
                        y error plus = stddev,
                        y error minus = stddev,
                    ] {
n, mean, stddev
256, 0.802593,0.000000
512, 0.791112,0.000000
1024, 0.772590,0.000000
2048, 0.816295,0.000000
4096, 0.761112,0.000000
8192, 0.778493,0.000000
16384, 0.789671,0.000000
32768, 0.708927,0.000000
65536, 0.666882,0.000000
131072, 0.726953,0.000000
262144, 0.757113,0.033682
524288, 0.723203,0.077599
1048576, 0.759792,0.014314
};

                \addplot+ [
                    error bars/.cd,
                        y explicit,
                        y dir=both,
                    ] table [
                        x = n,
                        y = mean,
                        y error plus = stddev,
                        y error minus = stddev,
                    ] {
n, mean, stddev
256, 0.765667,0.000000
512, 0.694332,0.000000
1024, 0.717336,0.000000
2048, 0.732997,0.000000
4096, 0.768333,0.000000
8192, 0.753977,0.000000
16384, 0.832716,0.000000
32768, 0.772681,0.008802
65536, 0.736027,0.000000
131072, 0.806973,0.000000
262144, 0.749486,0.000000
524288, 0.778895,0.034229
1048576, 0.753318,0.043248
};

                \legend{$k = 1$, $k = 2$, $k = 3$, $k = 4$, $k = 5$, $k = 6$, 
                        $k = 7$, $k = 8$, $k = 9$, $k = 10$}
            \end{semilogxaxis}
        \end{tikzpicture}
    \end{adjustbox}
    \caption{Monte-Carlo simulation of $F_n$, using $\gamma_6$ as unit.}
    \label{fig:monte_carlo}
\end{figure}

\section{Numerical Stability of Quickhull}

Now that we understand the accuracy of the building blocks of 
Algorithm~\ref{alg:quickhull_float}, we are ready to find bounds on the forward
error. We can simplify the analysis by making the following observation.
The first step of the convex hull algorithm finds the left-most and right-most
points by using coordinate comparison, which is exact. After that the structure
is the same as \texttt{HULL}. For this reason we restrict the analysis to the
recursive function \texttt{HULL}.

\begin{theorem}
    The forward error of Quickhull is bounded by $2DF_{|P|}$,
    where $D$ is the depth of the recursion and $F_{|P|}$ is a bound of
    Lemma~\ref{lem:farther_reduc}.
\end{theorem}

\begin{proofs}
    We fix each point where $\widehat{rt}$ or $\widehat{frt}$ made an incorrect
    decision by moving them a little bit, using Lemma~\ref{lem:right-turn} and 
    Lemma~\ref{lem:farther_reduc}. This yields $\tilde{P} \approx P$, 
    $\tilde{S}_1 \approx S_1$, and $\tilde{S}_2 \approx S_2$ such that
    $CH(\tilde{P} \cup \{p, q\}) = \{p\} \cup CH(\tilde{S}_1) \cup \{r\} \cup 
                     CH(\tilde{S}_2) \cup \{q\}$. We can push this backward 
    error $d_M(P, \tilde{P})$ forward using the perturbation analysis. Combining
    this with an inductive argument on the recursion yields the following chain
    of approximations.

    \begin{tikzcd}
        {\hat{H} = \{p\} \cup \hat{H}_1 \cup \{r\} \cup \hat{H}_2 \cup \{q\}}
        \arrow[d, leftrightarrow, "\approx\ \text{(induction hypothesis)}"] \\
        {\{p\} \cup CH(S_1) \cup \{r\} \cup CH(S_2) \cup \{q\}}
        \arrow[d, leftrightarrow, "\approx\ \text{(perturbation analysis)}"] \\
        {\{p\} \cup CH(\tilde{S}_1) \cup \{r\} \cup CH(\tilde{S}_2) \cup \{q\}}
        \arrow[d, leftrightarrow, "=\text{ (construction of }\tilde{P})"] \\
        {CH(\tilde{P} \cup \{p, q\})}
        \arrow[d, leftrightarrow, "\approx\ \text{(perturbation analysis)}"] \\
        {CH(P \cup \{p, q\})}
    \end{tikzcd}
\end{proofs}

We can make the $\approx$ rigorous with triangle inequalities, which yields
the following proof.

\begin{proof}
    We construct $\tilde{P}$ as follows.
    By Lemma~\ref{lem:right-turn}, all points that are classified incorrectly 
    have a distance of at most $\gamma_6 M$ to either $pr$ or $rq$. If 
    $u \in P$ has been incorrectly added $S_1$, we move it
    $\frac{1}{2}(d(u, pr) + \gamma_6 M)$ perpendicular to $pr$. This yields a 
    point $\hat{u}$ that satisfies $rt(p, \hat{u}, r)$ and 
    $d(u, \hat{u}) \leq \gamma_6 M$ by Lemma~\ref{lem:right-turn}. We do the same 
    for points that have been incorrectly added to $S_2$. Points that were 
    incorrectly discarded are moved to $pr$ (or equivalently $rq$). 
    If $r_1$ is not actually farthest from $pr$, we only have to move it at
    most $F_{|S_1|} / M$ farther from $pr$ to make it the farthest point.
    The points $p, r, q$ are not moved as
    $\widehat{rt}$ correctly classifies them. 
    This gives us $\tilde{S}_1, \tilde{S}_2 \subseteq \tilde{P}$ such 
    that $d_M(\tilde{S}_i, S_i), d_M(\tilde{P}, P) \leq F_{|S_i|}$ and
    $CH(\tilde{P} \cup \{p, q\}) = 
    \{p\} \cup CH(\tilde{S}_1) \cup \{r\} \cup CH(\tilde{S}_2) \cup \{q\}$.

    We now do our induction. The base case is trivial. The induction hypothesis
    gives us $\hat{H}_i$ such that
    $d_M(\hat{H}_i, CH(S_i)) \leq 2(D - 1) F_{|P|}$.
    As Hausdorff distance takes maxima, this implies that also

    \begin{equation}\label{eq:1}
        d_M(\{p\} \cup \hat{H}_1 \cup \{r\} \cup \hat{H}_2 \cup \{q\}, 
            \{p\} \cup CH(S_1) \cup \{r\} \cup CH(S_2) \cup \{q\})
            \leq 2(D - 1) F_{|P|}
    \end{equation}
    By our perturbation analysis $d_M(CH(\tilde{S}_i), CH(S_i)) \leq F_{|S_i|}$,
    so 
    \begin{equation}\label{eq:2}
    \begin{aligned}
        d_M(\{p\} \cup CH(S_1) \cup \{r\} \cup CH(S_2) \cup \{q\}), \\
            \{p\} \cup CH(\tilde{S}_1) \cup \{r\} \cup CH(\tilde{S}_2) \cup \{q\}) \\
        \leq \max(F_{|S_1|}, F_{|S_2|}) \leq F_{|P|}.
    \end{aligned}
    \end{equation}
    By construction of $\tilde{P}, \tilde{S}_1, \tilde{S}_2$, we have
    \begin{equation}\label{eq:3}
    CH(\tilde{P} \cup \{p, q\}) = 
    \{p\} \cup CH(\tilde{S}_1) \cup \{r\} \cup CH(\tilde{S}_2) \cup \{q\},
    \end{equation}
    and by perturbation analysis
    \begin{equation}\label{eq:4}
        d_M(CH(\tilde{P} \cup \{p, q\}), CH(P \cup \{p, q\}))  
        \leq \max(F_{|S_1|}, F_{|S_2|}) \leq F_{|P|}.
    \end{equation}
    The induction step now follows from combining equations 
    \ref{eq:1}, \ref{eq:2}, \ref{eq:3}, \ref{eq:4} using the triangle 
    inequality.
\end{proof}

\subsection{Discussion}

The worst-case error bound is poor, $O(|P|^2 \epsilon)$, if 
Algorithm~\ref{alg:farthest1} is used
for the reduction, and if the recursion has depth $O(|P|)$. However, a recursion
depth of $O(|P|)$ gives Quickhull a runtime complexity of $O(|P|^2)$, making
it useless on that input anyway. If the recursion is reasonably 
balanced, we have $D \in O(\log(|CH(P)|))$. 
Under Hypothesis~\ref{hyp:farther_reduc_random} we have also have a
far tighter expected bound $\mathbb{E}[F_n] \in O(\epsilon)$. Under these
assumptions, which we believe reasonable, the forward error bound will be
$O(\log(|CH(P)|)\epsilon)$.

If we do not want to assume Hypothesis~\ref{hyp:farther_reduc_random} holds, 
we can choose $m \approx \sqrt{n}$ in Algorithm~\ref{alg:farthest2} to reduce 
the bound to $O(\sqrt{|P|}D \epsilon)$. Algorithm~\ref{alg:farthest3} further 
reduces the bound to $O(\log|P|D \epsilon)$, but this significantly complicates 
the implementation. If the runtime is reasonable, so $D \in O(\log|P|)$, we get 
the bounds promised in the introduction.

We also note that a parallel implementation of Quickhull will lead to
Algorithm~\ref{alg:farthest2}, for $m$ the number of processors, so
parallelisation will improve, rather than worsen, the forward error bound.

\section{Related and Future Work}

Finding planar convex hulls is a well-studied problem, and there exist at 
least ten algorithms for solving it \cite{Graham72, Jarvis73, Eddy77, 
Preparata77, Bykat78, Akl78, Andrew79, Clarkson93, Barber96, Chan96}. 

As far as we are aware, robustness work of the convex hull problem has all been
done for variations of Graham Scan. Fortune~\cite{Fortune89} modifies 
Graham Scan to obtain an algorithm with forward error bound of $O(|P| \epsilon)$.
Assuming Quickhull has a runtime of $O(|P| \log |P|)$ we obtain slightly worse
error bound $O(|P| \log(|P|) \epsilon)$ when using the straightforward 
reduction algorithm, or a better one when using more complicated reduction 
algorithms. In practice, we expect an even
tighter error bound $O\left(\log(|P|)|\log(|CH(P)|) \epsilon\right)$. 
For Graham-Fortune
the bound will not be much tighter in practice, which we can see as follows. 
This algorithm first sorts the
points by $x$-coordinate and adds the points to a stack in that order. 
The stack is then popped until the new point makes a right turn with the two 
points underneath it. The error is linear in the number of successive pops. 
There are $|P| - |CH(P)|$ points are popped in total, so at least
\[
    \frac{|P|}{|CH(P)|} - 1
\]
points must be popped successively which is large for $|P| \gg |CH(P)|$.

Jaromczyk et al.~\cite{Jaromczyk94} modify Fortune's algorithm further to a 
variant that is backward stable with error bound $O(|P| \epsilon)$. Our 
perturbation analysis shows that this also implies a forward error bound of 
$O(|P| \epsilon)$. Quickhull does not necessarily have a small backward
error because the set we compute may not 
be convex, and therefore not the solution to any problem. This may be 
problematic when the hull we compute becomes the input of a different algorithm 
that expects convex input. In fact, even an output that is truly convex is 
problematic when we cannot test this at our level of accuracy. To that end
Li et al.~\cite{Li90} propose computing strongly convex hulls, where we reject
points if their adjacent sides are almost \qty{180}{\degree}. They provide an 
algorithm that can construct a strongly convex hull from an approximate one
in linear time, which means it can be efficiently combined with Quickhull. 

Jiang et al.~\cite{Jiang06} provide a perturbation analysis for the convex
hull problem. They do not scale by $M$ and consider solutions to have an 
infinite error if they are not a minimal representation. We have chosen to
scale because we want to be invariant to units. If a minimal representation
is desired, we believe a strongly convex hull should be extracted from 
Quickhull's output using Li's algorithm. This adds minimal $O(|CH(P)|)$ 
overhead and reduces the representation further.

The analysis of reduction algorithms is similar to that of summation 
\cite{Higham93}, except in our case some algorithms are uniformly better
than others.

\section{Conclusion}

We have defined a distance $d_M$ to measure the error between the convex hull
and a computed approximation. The perturbation analysis with respect to this
distance shows that the convex hull problem is well-conditioned.

We recommend implementing the geometric test 
$d(u, pq) > d(u', pq)$ as $(q_y - p_y)(u_x - u'_x) < (q_x - p_x)(u_y - u'_y)$
and show that this test is accurate. Error in the Quickhull algorithm
as a whole is accumulated through recursion and finding the point $r$ that
is farthest from a line segment $pq$. This gives a worst-case bound on the
forward error of $O(|P|^2 \epsilon)$, where $\epsilon$ is the machine precision.
However, we argue that Quickhull would never be used on data where the 
recursion is deeper than $O(\log|P|)$, as Fortune-Graham would have superior 
runtime. Similarly, we do not expect error to accumulate when finding the 
farthest point $r$ in practice. This gives a more realistic forward error of 
$O\left(\log(|CH(P)|)\epsilon\right)$.

Though we do not believe this is useful in practice, we do provide algorithmic
variations that reduce the worst-case bound. By using a different reduction
algorithm for finding $r$, we can reduce the error bound
to $O(\sqrt{|P|}D\epsilon)$ or $O(D\log(|P|)\epsilon)$, for $D$ the
recursion depth. The depth $D$ is at most $|P|$, and $O(\log|P|)$ for data
where the runtime complexity is competitive to Graham-Fortune.

\paragraph{Acknowledgements}

We would like to thank Laura Scarabosio for the helpful 
discussion on numerical robustness for giving feedback on a draft of this 
chapter.

\bibliographystyle{plain}
\bibliography{paper}

\end{document}